\documentclass[12pt]{article}

\pdfoutput=1
\usepackage{amsmath,amssymb,amsthm,color,epsfig,mathrsfs,amsbsy}
\usepackage{graphicx}
\usepackage{psfrag}
\newtheorem{theorem}{\bf Theorem}

\newtheorem{condition}[theorem]{\bf Condition}
\newtheorem{problem}[theorem]{\bf Problem}

\newcommand\ZZ{\mathbb{Z}}
\newcommand\R{\mathcal{R}}
\renewcommand\P{\mathcal{P}}
\newcommand\ep{\epsilon}

\newcommand{\Honek}{H^1_\kappa(\Omega_R)}

\begin{document}

\def\R{\mathcal{R}}
\def\S{\mathcal{S}}
\def\P{\mathcal{P}}
\def\ep{\epsilon}

\begin{center}
{\bf \Large Wave Scattering and Guided Modes}
\\ \vspace{1ex}
{\bf \Large in Periodic Pillars}
\end{center}

\vspace{0.2ex}

\begin{center}
{\scshape \large Hairui Tu}
\end{center}

\begin{center}
{\itshape Department of Mathematics, Louisiana State University\\
Baton Rouge, LA \ 70803}
\end{center}

\vspace{3ex}
\centerline{\parbox{0.9\textwidth}{
{\bf Abstract.}\
We investigate the scattering of scalar harmonic source fields by a periodic
pillar, that is, a spatial structure that is periodic in one dimension and
of finite extent in the other two.
Uniqueness of scattering solutions can be abstracted by guided modes.
Extending results for periodic slabs to pillars, we give conditions
under which ``inverse'' pillars cannot admit guided modes.
In addition, we present a new construction of guided modes at frequency $\omega$
and Bloch wavenumber $\kappa$ with $\omega$ embedded in the continuous spectrum
for each $\kappa$.
These guided modes have period larger than the period of the structure,
and possess a real dispersion relation
$\omega=\omega(\kappa)$, which is atypical of modes at embedded frequencies.
}}

\vspace{3ex}
\noindent
\begin{mbox}

\end{mbox}

\vspace{3ex}
\hrule
\vspace{2ex}

\section{Introduction}
Plane-wave scattering and guided modes are important problems in the
study of photonic crystals and periodic structures.
Guided scalar modes are Helmholtz fields exponentially trapped within periodic structures
in the absence of source fields originating from the exterior, and this special feature
has lead to many applications of photonic crystals such as waveguides and light filters.
%
%
In the case of a periodic slab finite in one direction
and periodic in the other two directions,
variational techniques such as in \cite{GilbargTrudinger1998}\cite{Jost2002}
have been applied to analyze guided modes.
In \cite{Bonnet-BeStarling1994}, Bonnet-Bendhia and Starling formulate
the scattering problem and guided mode problem for periodic slabs
and prove the existence and nonexistence of guided modes,
including guided modes with frequencies
embedded in the essential spectrum of the corresponding self-adjoint operator.
In \cite{ShipmanVolkov2007}, 
Shipman and Volkov
provide a proof of the nonexistence of guided modes in inverse structures for photonic crystal
slabs, i.e., structures that have higher wave speed in the slabs than in the exterior.
In this paper, we present a systematic mathematical framework for
composite structures  that are periodic in one direction and
finite in the other two.
We call such structures ``periodic pillars''.

In order to derive the variational formulation, one needs to propose
the radiation condition through the Dirichlet-to-Neumann operator.
In \cite{Bonnet-BeStarling1994}, the radiation condition is enforced by mapping the Dirichlet
boundary data of the Fourier expansion 
to the Neumann boundary data.
Our approach employs Bessel functions (see \cite{Watson1944}) to do this, thanks to the fact that,
in the exterior of pillars, general Helmholtz fields
can be expanded as an infinite superposition of Fourier harmonics with Bessel
functions as coefficients.
%
%
We propose the radiation condition for periodic pillars and formulate the scattering
problem through standard variational techniques shown in \cite{Evans1998}\cite{GilbargTrudinger1998}\cite{Jost2002}.
Naturally, we prove the existence of the solutions to the plane-wave scattering problem, and
characterize the frequency estimations of guided modes using the functional-analytic framework
for structures not necessarily piecewise.
%

The main original proof in this paper is the existence of guided modes
with frequency above the cutoff frequency, that is, above the light cone
in the first Brillouin zone.
The existence of guided modes with the frequencies below  cutoff
are well established in \cite{Bonnet-BeStarling1994}, and the pair of frequency and wavenumber
satisfies continuous dispersion relations
and an extension to periodic pillars is straightforward.
The proof of existence of certain embedded guided modes is in general difficult.
In the proof in section 5.2 of \cite{Bonnet-BeStarling1994}, the existence of embedded
guided modes in periodic slabs is proved.
%
However, this type of embedded guided mode can be viewed as non-embedded
because the periodicity of the mathematical construction is chosen to
be larger than the smallest period of the structure and the mode.
%
In fact, any non-embedded guided mode in a periodic structure can be treated
as embedded if one artificially chooses a larger period:
this has the effect of reducing the size of Brillouin zone
in wavenumber space so that the reduced Bloch wavenumber of
the mode now lies above the light cone.
%
%
Thus these guided modes are subject to continuous dispersion relations, as
they are non-embedded from the point of view of their prime period.
In this paper, we prove the existence of truly nontrivial embedded guided modes
in section 4.1.
The embedded guided modes are obtained by designing a wavenumber-dependent subspace
 of fields that is
invariant under the Helmholtz equation and tuning the material parameters so
that the eigenfunction inside the subspace is located in the regime where
all the propagating harmonics automatically vanish in this subspace.
The eigenfunctions in this subspace are automatically guided modes with
frequencies embedded in the continuous spectrum for the given wavenumber.
Our guided modes do not admit a smaller period, yet they persist with perturbations of
the wavenumbers.
The modes admit a continuous dispersion relation above the cutoff frequency.

%
Our proofs of nonexistence of guided modes are based on the ideas in
\cite{Bonnet-BeStarling1994}\cite{ShipmanVolkov2007}.
The first proof is an analogy of the proof in \cite{ShipmanVolkov2007} for piecewise constant inverse structures.
In that study, the proof of the nonexistence relies on a restriction on the
width of the slabs.
The restriction on the geometry of the structures is still needed in our proof
of nonexistence.
We do not know whether this restriction is necessary and leave it as an open problem.
The second nonexistence result is established by using radial monotonicity
of material parameters, in analogy with Theorem 3.5 of \cite{Bonnet-BeStarling1994} that involves
an appropriate Rellich identity.

\section{Media Structure and Scattering Problem}
\subsection{Pillar Structure and Radiation Condition}
%
Consider an infinitely long ``pillar-shaped'' structure whose material parameters are periodic in the $z$-direction
with period $2\pi$, that is,
\begin{equation}
\begin{gathered}
 \epsilon(x,y,z+2\pi)=\epsilon(x,y,z), \quad\quad \mu(x,y,z+2\pi)=\mu(x,y,z), \quad\forall x,y,z. \\
 \end{gathered}
\end{equation}
and bounded in the $x,y$ directions, illustrated by
Figure \ref{fig:pillar}.
\begin{figure}[htt]\label{fig:pillar}
 \centering
    \includegraphics*[viewport=80 230 680 720, scale=0.4]{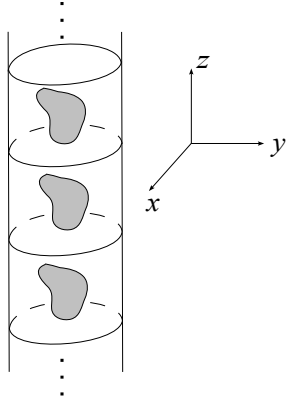}
    \caption{a pillar structure periodic in $z-$direction and bounded in $x,y-$direction}
 \end{figure}
Assume these coefficients to be bounded from below and above by positive numbers,
$\ep_-<\epsilon<\ep_+, \mu_-<\mu<\mu_+$.
We are interested in time-harmonic fields satisfying the scalar wave equation
that are spatially pseudo-periodic in $z$, that is $u(x,y,z)e^{-i\omega t}$,
where $u$ satisfies $u(x,y,z+2\pi)=u(x,y,z)e^{2\pi\kappa i}$
and  $\kappa$ is the \emph{Bloch wavenumber}.
The field $u(x,y,z)$ is governed by the Helmholtz equation
\begin{equation}\label{eqn:Helmholtz}
\nabla\cdot\frac{1}{\mu}\nabla u(x,y,z)+\epsilon\omega^2 u(x,y,z)=0,
\end{equation}
in which $\ep$ and $\mu$ are bounded positive 
and $2\pi$-periodic in $z$.
%

We use  $\Omega=\{(x,y,z):-\pi\le z\le\pi\}$ to denote one period.
Suppose that $\epsilon=\epsilon_0, \mu=\,u_0$ for $r>R$ and we denote the restricted domain
$\Omega_R=\{(x,y,z):-\pi\le z\le \pi, r=\sqrt{x^2+y^2}<R\}$,
whose boundary is $\Gamma_R=\{(x,y,z)\in\Omega: -\pi<x,y<\pi, r=R\}$ plus the upper
and lower disks.\\

The spatial factor of a time-harmonic acoustic or electromagnetic wave is governed by
the Helmholtz equation \eqref{eqn:Helmholtz}.
%
%
By the $\kappa$-pseudo-periodicity,  $u$ can be expanded as an infinite
superposition of Fourier harmonics in $\Omega\setminus\Omega_R$:
\begin{equation}\label{eqn:Expansion}
u(x,y,z)=\sum_{m=-\infty}^\infty\sum_{\ell=-\infty}^\infty R_{m,\ell}(r)e^{i\ell\theta}e^{i(m+\kappa)z},
\end{equation}
where
\begin{equation}
R_{m,\ell}(r)=\begin{cases}
a_{m\ell}H^1_\ell(\eta_m r)+b_{m\ell}H^2_\ell(\eta_m r),   & \mbox{if } \eta_m\neq0, \\
c_{m1}+c_{m2}\ln|r|,  & \mbox{if } \eta_m=0, \ell=0, \\
c_{m\ell1}|r|^\ell+c_{m\ell2}|r|^{-\ell},   & \mbox{if } \eta_m=0, \ell\neq 0.
\end{cases}
\end{equation}
In this expansion, if we assume $\eta_m>0$ when $\eta_m^2>0$ and $i\eta_m>0$ when $\eta_m^2$,
 the Hankel functions $H^1_\ell(\eta_m r)$ are outgoing or exponentially decaying,
depending on whether $\eta_m$ is imaginary or real,
as $r\rightarrow\infty$, and the Hankel functions $H^2_\ell(\eta_m r)$ are incoming or exponentially
growing.

\smallskip
The Sommerfeld radiation condition
$\displaystyle\lim_{r\rightarrow\infty}r^{1/2}(\frac{\partial}{\partial r}-i\kappa)u=0$ is
equivalent to the following conditions, which is required for the study of the
scattering problem by a periodic pillar.
\begin{condition}[Radiation condition]\label{cond:radiation}
A field $u(r,\theta,z)$ satisfies the radiation condition if it admits the following Fourier-Bessel
representation for $r>R$:
\begin{equation}
\begin{split}
u(r,\theta,z)
   &= \sum_{m\in \mathcal{Z}_p\cup\mathcal{Z}_e}\sum_{\ell\in\mathbb{Z}} a_{m\ell} H^1_\ell(\eta_m r) e^{i\ell\theta}e^{i(m+\kappa)z}\\
  & +\sum_{m\in \mathcal{Z}_a}\left[\sum_{\ell>0}c_{ml_2}|r|^{-\ell}e^{i\ell\theta}
   +\sum_{\ell<0}c_{m\ell_1}|r|^\ell e^{i\ell\theta}\right] e^{i(m+\kappa)z}
\end{split}
\end{equation}
where the subsets $\mathcal{Z}_{p,a,e}$ of $\ZZ$ depend on $\kappa$ and are defined by
\begin{eqnarray*}
m\in \mathcal{Z}_p\Leftrightarrow \eta_m^2>0, \eta_m>0 \mbox{ (propagating harmonics)} \\
m\in \mathcal{Z}_a\Leftrightarrow \eta_m^2=0, \eta_m=0 \mbox{ (algebraic harmonics)} \\
m\in \mathcal{Z}_e\Leftrightarrow \eta_m^2<0, -i\eta_m>0 \mbox{ (evanescent harmonics) }.
\end{eqnarray*}
\end{condition}

\subsection{Scattering Problems}
Although this paper deals with guided modes, it would be incomplete without a discussion
of the scattering problem. In fact, the nonexistence of guided modes is equivalent to the
unique solvability of the scattering problem.
Now consider the specific case of a plane wave.
\begin{problem}[Scattering problem, strong form]\label{prob:scattering}
Given $\epsilon_0,\epsilon_1,\mu_0,\mu_1>0$, find $u$ on $\Omega$ such that
\begin{equation}\left\{
\begin{split}
\nabla&\cdot\frac{1}{\mu}\nabla u+\ep\omega^2u=0   \mbox{ in } \Omega,\\
u \mbox{  is}& \mbox{ continuous on }  \partial \Omega,   \\
\frac{1}{\mu}&\frac{\partial u}{\partial n} \mbox{ is continuous on } \partial\Omega, \\
u^{inc}&=\sum_{m\in\mathcal{Z}_p}u_m^{inc}e^{i(\kappa_1x+\kappa_2y+(m+\kappa)z)},\\
u^{sc}&=u-u^{inc}   \mbox{ and its derivatives are }\kappa \mbox{-periodic in }  z,\\
u^{sc}&=u-u^{inc} \mbox{ satisfies the radiation condition }.
\end{split}
\right.
\end{equation}
\end{problem}

An incident plane wave $e^{i(\kappa_1x+\kappa_2y+\kappa_3z)}$
satisfies the Helmholtz equation exterior to the pillar, and thus
its wave vector satisfies $\kappa_1^2+\kappa_2^2=\eta_m^2$.
We take $\kappa$ to lie in the first Brillouin zone $[-1/2,1/2)$.
Noting that the function $e^{\frac{Z}{2}(t-\frac{1}{t})}$ generates
the Bessel functions, i.e.,  $e^{\frac{Z}{2}(t-\frac{1}{t})}=\sum_\ell t^\ell J_\ell(Z)$,
we let $t=e^{i(\theta+\theta_0)}$ to obtain
$e^{iZ\sin(\theta+\theta_0)}=\Sigma_\ell J_\ell(Z)e^{i\ell(\theta+\theta_0)}$.
Then with $\sin\theta_0=\frac{\kappa_1}{\eta_m}, \cos\theta_0=\frac{\kappa_2}{\eta_m}$,
and $Z=\eta_m r$, the incident wave can be written as a superposition
of Hankel functions:
\begin{equation}\label{idty:Incident}
\begin{split}
e^{i(\kappa_1x+\kappa_2y+\kappa_3z)}=&
   e^{i(\eta_m r\cos\theta\sin\theta_0+\eta_m r\sin\theta\cos\theta_0)}e^{i\kappa_3z}\\
  =&e^{i\eta_m r\sin(\theta+\theta_0)}e^{i(m+\kappa)z}\\
  =&\sum_{\ell\in\ZZ} J_\ell(\eta_m r)e^{i\ell(\theta+\theta_0)}e^{i(m+\kappa)z}\\
  =&\sum_{\ell\in\ZZ} \frac{1}{2}\left[H^1_\ell(\eta_m r)+H^2_\ell(\eta_m r)\right]e^{i\ell(\theta+\theta_0)}e^{i(m+\kappa)z}.
\end{split}
\end{equation}
As a result, the scattering problem of plane waves can be reduced
to the linear superposition of propagating Fourier harmonics with Hankel functions.


To analyze the solvability of the scattering problem, it is convenient to reduce
to the truncated domain $\Omega_R$.
Define the pseudo-periodic field space
$H^1_\kappa(\Omega_R)=\{u\in H^1(\Omega_R):u(x, y, \pi)=u(x,y,-\pi)e^{2\pi\kappa i}\}$.
On the boundary $\Gamma_R$, the radiation condition
is characterized by a Dirichlet-to-Neumann map
$T:H^{\frac{1}{2}}_\kappa(\Gamma_R)\rightarrow H^{-\frac{1}{2}}_\kappa(\Gamma_R)$
(as in the Definition 5.19 of \cite{CakoniColton2006})
\begin{equation}
T:\sum_{m,\ell}\hat{u}_{m\ell}e^{i\ell\theta}e^{i(m+\kappa)z} \mapsto
   \sum_{m,\ell}\gamma_{m\ell}\hat{u}_{m\ell}e^{i\ell\theta}e^{i(m+\kappa)z},
\end{equation}
where
\[
\gamma_{m\ell}=
\begin{cases}
\frac{-\eta_m H^{1'}_\ell(\eta_mR)}{H^1_\ell(\eta_mR)}, &\mbox{ if } m\not\in \mathcal{Z}_a, \\
|\ell|R^{-1}, &\mbox{ if } m\in\mathcal{Z}_a \mbox{ and }\ell\neq0, \\
0, &\mbox{ if }  m\in\mathcal{Z}_a \mbox{ and }\ell=0.
\end{cases}
\]
To satisfy the radiation condition, the harmonics in \eqref{eqn:Expansion} with
$H^2_\ell(\eta_m r)$ for $m\in \mathcal{Z}_p\cup\mathcal{Z}_e$, harmonics
$(c_{m1}+c_{m2}\ln|r|) e^{i\ell\theta}e^{i(m+k)z}$ for $m\in\mathcal{Z}_a, \ell=0$, and harmonics
with $|r|^\ell$ for $m\in\mathcal{Z}_a, \ell>0$ all vanish.
The radiation condition is hence enforced by
\begin{equation}
\partial_n u+Tu=0 \mbox{ on } \Gamma_R.
\end{equation}
The operator $T$ has a ``nonnegative evanescent'' part $T_{e}$ and a ``propagating'' part $T_{p}$:
\begin{equation}
T=T_{e}+T_{p}.
\end{equation}
\begin{equation}
\widehat{(T_{e}f)}_{m\ell}=
\begin{cases}
\frac{-\eta_m H^{1'}_\ell(\eta_mR)}{H^1_\ell(\eta_mR)}\hat f_{m\ell},  \mbox{ if } m\in \mathcal{Z}_e, \\
|\ell|R^{-1}\hat f_{m\ell},  \mbox{ if } m\in\mathcal{Z}_a \mbox{ and } \ell\neq 0 ,\\
0, \mbox{ otherwise },
\end{cases}
\end{equation}
\begin{equation}
\widehat{(T_{p}f)}_{m\ell}=
\begin{cases}
\frac{-\eta_m H^{1'}_\ell(\eta_mR)}{H^1_\ell(\eta_mR)}\hat f_{m\ell},  \mbox{ if } m\in \mathcal{Z}_p, \\
0, \mbox{ otherwise }.
\end{cases}
\end{equation}

The variational form of the scattering problem in the truncated domain is
\begin{problem}[Scattering problem, variational form]\label{prob:weak1}
\begin{equation}
\left\{
\begin{split}
& u\in H^1_\kappa(\Omega_R)  \\
& a(u,v)-\omega^2b(u,v)=f(v),  \forall v\in H^1_\kappa(\Omega_R)
\end{split}
\right.
\end{equation}
where
\[
\begin{split}
a(u,v)&=\int_{\Omega_R}\frac{1}{\mu}\nabla u\cdot\nabla \bar v+\frac{1}{\mu_0}\int_{\Gamma_R}(Tu)\bar v,\\
b(u,v)&=\int_{\Omega_R}\ep u \bar v,\\
f(v)&=\frac{1}{\mu_0}\int_{\Gamma_R}\left[(\partial_n u^{inc}+Tu^{inc})\bar v\right].\\
\end{split}
\]
\end{problem}

%
The variational form in problem \ref{prob:weak1} can be written as
\[
a(u,v)-\omega^2b(u,v)=c_1(u,v)+c_2(u,v)
\]
with
$c_1(u,v)=\int_{\Omega_R}(\frac{1}{\mu}\nabla u\cdot\nabla \bar v+\ep u \bar v )+\frac{1}{\mu_0}\int_{\Gamma_R}(Tu)\bar v$
and $c_2(u,v)=-\ep(\omega^2+1)\int_{\Omega_R}u\bar v$.
Define operators $C_1$ and $C_2$ on $H^1_\kappa(\Omega_R)$ by
$(C_1u,v)_{H^1_\kappa(\Omega_R)}=c_1(u,v)$ and
$(C_2u,v)_{H^1_\kappa(\Omega_R)}=c_2(u,v)$.
Because of the coercivity of $c_1$ and the compact embedding of $L^2(\Omega)$ into $\Honek$,
we know that the operator $C_1$ is an automorphism and $C_2$ is compact.

If we denote by $f^{inc}$ the unique element of $H^1_\kappa(\Omega_R)$ such that
$(f^{inc},v)_{H^1_\kappa(\Omega_R)}=f(v)$,
the variational form of the scattering problem can be characterized by the following
operator form
\[
(C_1u,v)+(C_2u,v)=(f^{inc},v), \forall v\in\Honek,
\]
i.e.
\[
C_1u+C_2u=f^{inc}.
\]
The Fredholm alternative theory implies that the nonuniqueness of the solution
of this problem is equivalent to the singularity of the corresponding homogeneous
problem $C_1u+C_2u=0$, whose weak form is given by
\begin{equation}\label{eqn:weakhomogeneous}
a(u,v)-\omega^2 b(u,v)=0, \forall v\in H^1_\kappa(\Omega_R).
\end{equation}

\begin{theorem}
The scattering problem has at least one solution, and the set of solutions is at most finite
dimensional.
\end{theorem}
\begin{proof}
From equation \eqref{idty:Incident}, we can express the incident plane wave as a superposition of
harmonics
$\displaystyle\sum_\ell\frac{1}{2}\left[H^1_\ell(\eta_m r)+H^2_\ell(\eta_m r)\right]e^{i\ell(\theta+\theta_0)}e^{i(m+\kappa)z}$,
with $m\in\mathcal{Z}_p$.
By the Fredholm alternative, the scattering problem has a solution if and only if
\[
(f^{inc},w)=0, \mbox{ for all } w\in\mbox{Null}(C_1+C_2)^\dagger,
\]
 i.e. for all $w$ such that
\[
a(v,w)-\omega^2 b(v,w)=0, \forall v\in H^1_\kappa(\Omega_R)
\]
This $w$ satisfies
\[
\overline{a(w,v)}-\omega^2 b(w,v)=0, \forall v\in H^1_\kappa(\Omega_R)
\]
and by the decomposition of $T$, we know that for all $m\in\mathcal{Z}_p$,
$\hat{w}_m=0$.
By the definition of $f^{inc}$, showing $(f^{inc},w)=0$ is equivalent to
showing that $\int_{\Gamma_R}(\partial_n+T)u^{inc}\bar w=0$.
This is satisfied by the function $w$ above.

The space of solutions is finite dimensional because $C_1$ is invertible
and $C_2$ is compact.
%
\end{proof}

\section{Guided Modes}

A guided mode is a solution to the Helmholtz equation in the periodic domain
in the absence of any source field.
In the weak form
, it is a solution to the homogeneous equation (\ref{eqn:weakhomogeneous}).

The sesquilinear form
\[
a^\omega(u,v)=\int_{\Omega_R}\frac{1}{\mu}\nabla u\cdot\nabla\bar v+\frac{1}{\mu_0}\int_{\Gamma_R}(T^\omega u)\bar v
\]
can be split into evanescent and propagating parts,
\[
a^\omega_e(u,v)=\int_{\Omega_R}\frac{1}{\mu}\nabla u\cdot\nabla\bar v+\frac{1}{\mu_0}\int_{\Gamma_R}(T^\omega_eu)\bar v,
\]
\[
a^\omega_p(u,v)=\frac{1}{\mu_0}\int_{\Gamma_R}(T^\omega_pu)\bar v.
\]
If the frequency and the wavenumber are assumed to be real, the  form $a^\omega_e(u,v)$ is Hermitian,
and we have the following theorem.
%

\begin{theorem}(Real eigenvalues)
If the frequency $\omega$ is real, then $u\in H^1_\kappa(\Omega_R)$ solves the equation
(\ref{eqn:weakhomogeneous}) if and only if
\begin{equation}\label{eqn:WeakReal1}
a_e^\omega(u,v)-a_p^\omega(u,v)-\omega^2b(u,v)=0,\forall v\in H^1_\kappa(\Omega_R)
\end{equation}
and if and only if
\begin{equation}\label{eqn:WeakReal2}
\begin{cases}
a_e^\omega(u,v)-\omega^2 b(u,v)=0, \forall v\in H^1_\kappa(\Omega_R), \\
\widehat{(u|_{\Gamma_R})}_m=0, \forall m\in\mathcal{Z}_p.
\end{cases}
\end{equation}
\end{theorem}

The eigenfrequencies can be obtained by applying the min-max principle to the real form
in \eqref{eqn:WeakReal2}.
When $\omega<\sqrt{\frac{\kappa^2}{\epsilon_0\mu_0}}$, the
solutions $u$ of $a_e^\omega(u,v)-\omega^2 b(u,v)=0, \forall v\in H^1_\kappa(\Omega_R)$
are guided modes 
since this regime admits no propagating harmonics and hence the second conditions
in \eqref{eqn:WeakReal2} are automatically satisfied.
When $\omega\ge\sqrt{\frac{\kappa^2}{\epsilon_0\mu_0}}$, to be guided modes, these solutions $u$
must satisfy the extra conditions $\widehat{(u|_{\Gamma_R})}_m=0, \forall m\in\mathcal{Z}_p$
where $\mathcal{Z}_p$ is nonempty.
We will design some periodic structures that admit guided modes in the next section.
We have the following theorem on properties of the frequencies.
The proof is similar to that for periodic slabs, for which one may refer to
\cite{ShipmanVolkov2007}\cite{Bonnet-BeStarling1994}.

\begin{theorem}(Eigenvalues and Frequencies)\label{thm:EigenPillar}
The problem $a_e^\omega(u,v)-\lambda b(u,v)=0, \forall v\in H^1_\kappa(\Omega_R)$ has a
nondecreasing sequence of eigenvalues
$\{\lambda_j\}^\infty_{j=1} $,
obtained through the minmax principle,
\begin{equation}
\lambda_j=\sup_{\dim V=j-1,V\subset\Honek} \inf_{u\in V^\bot\setminus0}\frac{a_e(u,u)}{b(u,u)},
\end{equation}
which tend to $+\infty$ as $j\rightarrow\infty$.
Moreover, the homogeneous problem
$a^\omega_e(u,v)-\omega^2 b(u,v)=0, \forall v\in H^1_\kappa(\Omega_R)$ has a nontrivial
solution if and only if $\omega^2=\lambda_j(\omega)$.
This frequency can be denoted by $\omega_j$.

%
%
\end{theorem}

The problem of scattering and guided modes can be posed equivalently through a positive
form in the entire strip $\Omega$.
This formulation leads to the determination of the spectrum of the scattering problem, and
in particular, the eigenvalues below the continuous spectrum.
However, the formulation on the truncated domain $\Omega_R$ is necessary for the determination
of the embedded eigenvalues, which will be our focus in the next section.

We can derive the weak form of the guided modes problem:
\begin{equation}
a_S(u,v)=\omega^2b_S(u,v), \forall v\in H^1_\kappa(\Omega)
\end{equation}
where
\begin{eqnarray}
a_S(u,v)=\int_\Omega\frac{1}{\mu}\nabla u\cdot\nabla\bar v,  \\
b_S(u,v)=\int_\Omega\ep u\bar v.
\end{eqnarray}
The associated operator is the unbounded operator
\begin{equation}
S_\kappa u=-\frac{1}{\ep}\nabla\cdot\frac{1}{\mu}\nabla u.
\end{equation}
It is defined on the domain
$D(S_\kappa)=\{u\in H^1_\kappa(\Omega):
\mbox{ There exists a } C \mbox{ such that }  |a_S(u,v)|\leq C\sqrt{b_S(v,v)}, \forall v\in H^1_\kappa(\Omega)\}$.
This operator is positive self-adjoint and its eigenvectors and eigenvalues
are solutions of the guided modes problem. We denote the spectrum of $S_\kappa$ as $\sigma$, and
its essential spectrum as $\sigma_{ess}$.
The following theorem is an adaptation of Theorem 4.1 of \cite{Bonnet-BeStarling1994} to periodic pillars.
\begin{theorem}
i) $\sigma\subset [\frac{\kappa^2}{\mu_+\ep_+},+\infty)$, where $\mu_+=\sup_{\Omega}\mu, \ep_+=\sup_{\Omega}\ep$; \\
ii) $\sigma^{ess}=[\frac{\kappa^2}{\mu_0\ep_0},+\infty)$;  \\
iii)  there are finitely many eigenvalues $\tilde\lambda_j(\kappa)$ below $\frac{\kappa^2}{\mu_0\ep_0}$, and
 $\{\tilde\lambda_j(\kappa)\}$ is an increasing sequence that converges to $\frac{\kappa^2}{\mu_0\ep_0}$.
\end{theorem}

\section{Existence and Nonexistence of Guided Modes}
\subsection{Existence}
The focus of this section is to find guided modes with frequency $\omega$ such
that $\omega^2$ is embedded in the continuous spectrum of $S_\kappa$.
As discussed in the previous section, certain extra conditions should be
satisfied and hence bring the difficulty.

In \cite{Bonnet-BeStarling1994}, guided modes are proved to exist in a symmetric structure and
a periodic slab with a finer periodicity.
The idea is to consider a closed subspace $F$ on which the operator $S_\kappa$ has
a cutoff frequency that is greater the cutoff frequency on $\Honek$, and prove
the existence of guided modes corresponding to eigenfrequencies lying between
 these two cutoff frequencies.
 These eigenfunctions are automatically guided modes lying
 in $F$ because their frequencies are below the cutoff frequency,
 but the frequencies are embedded in the essential spectrum of $S_\kappa$ for $\Honek$.
In their proof, the embedded guide modes retain the original pseudo-periodicity, but they
are simply non-embedded guided modes with a smaller pseudo-period.
By artificially choosing a larger period, any guided modes with frequencies below
the cutoff frequency can be seen as embedded guided modes in the same
structure with the larger period.
%
%
In this paper, we present a proof of the existence of non-artificial 
guided modes with frequencies embedded in the essential spectrum the operator $S_\kappa$.
%
We only need the parameters $\epsilon,\mu$ to have smaller period, but the guided modes
do not have smaller pseudo-period.
%
%

Our newly designed pillar is a periodic structure with period $\frac{2\pi}{L}$
for $L\ge2$ in $\mathbb{Z}$ that supports guided modes with pseudo-period
strictly greater than $\frac{2\pi}{L}$.
\begin{theorem}\label{thm:Existence}
For any $\kappa$ in the first Brillouin zone of the structure of period $2\pi$, there exists $\ep,\mu$ with period
$\frac{2\pi}{L}$ for $L\ge2$ that admits a guided mode with frequencies $\omega$ lying above
the cutoff frequency.
\end{theorem}
\begin{proof}
Write $u\in\Honek$ as a Fourier expansion $u(r,\theta,z)=\displaystyle\sum_{m}u_m(r,\theta)e^{i(m+\kappa)z}$.
Given $M,N\in\mathbb{N}$ with $2M+N+2=L$, define a nontrivial subspace of $H^1_k(\Omega)$:
\begin{equation}
V=\left\{u\in H^1_k(\Omega): u_m(r,\theta)\equiv0 \mbox{, if } |m-j(2M+N+2)|\le M \mbox{ for some } j\in\ZZ\right\}
\end{equation}
Therefore, for $-M+j(2M+N+2)\le m\le M+j(2M+N+2)$, the coefficients $u_m(r,\theta)$ are $0$, and for
$M+1+j(2M+N+2)\le m\le M+N+1+j(2M+N+2)$, the coefficients $u_m(r,\theta)$ are possibly nonzero.
%

We claim that $\epsilon V\subseteq V$, $\mu^{-1} V\subseteq V$.
In fact, let $(\epsilon)_m(r,\theta)$ be the Fourier coefficients of $\epsilon$.
The periodicity of the structure implies that $(\epsilon)_m(r,\theta)\equiv0, \forall r,\theta$,  except
when $m=j(2M+N+2)$ for some integer $j$.
For any $u\in V$, if $|m-j(2M+N+2)|\le M$ for some $j\in\ZZ$, we calculate the $m^{th}$
Fourier coefficient of $\epsilon u$:
\[
\begin{split}
(\epsilon u)_m & =\sum_\ell(\epsilon)_\ell u_{m-\ell} \\
          & =\sum_j (\epsilon)_{j(2M+N+2)} u_{m-j(2M+N+2)} \\
          & =0   \mbox{ ,  because } u_{m-j(2M+N+2)}=0 \mbox{ for the field } u\in V.
\end{split}
\]
Therefore, $\epsilon u\in V$.
Similarly, $\mu^{-1}V\subseteq V$.

Therefore, the subspace $V$ is also invariant under the operator $\nabla\cdot\frac{1}{\mu}\nabla$.
Thanks to the invariance properties, we can consider the Helmholtz equation in the subspace $V$.
The solution $u\in V$ to the weak formulation $a_r^\omega(u,v)-\omega^2b(u,v)=0, \forall v\in V$
is also a solution to $a_r^\omega(u,v)-\omega^2b(u,v)=0,\forall v\in\Honek$.
In fact, for any field $u\in V$ and $v\in V^{\perp}$, $\nabla\cdot\frac{1}{\mu}\nabla u+\omega^2\ep u\in V$
implies that $\nabla\cdot\frac{1}{\mu}\nabla u\bar v+\omega^2\ep u\bar v=0$ for all $v\in V^{\perp}$.
Integrating it we obtain
\[
\begin{split}
\int_{\Omega}\nabla\cdot\frac{1}{\mu}\nabla u\bar v+\int_\Omega^2\ep u\bar v &=
   \frac{1}{\mu_0}\int_{\Gamma_R}\partial_nu\bar v-\int_\Omega\frac{1}{\mu}\nabla u\cdot\nabla\bar v+\int_\Omega^2\ep u\bar v\\
   &=-a^\omega_r(u,v)+b(u,v)\\
   &=0.
\end{split}
\]
%
We can obtain a pair $(\omega,u)$ by applying
the min-max principle to the Rayleigh quotient $\frac{a_r(u,u)}{b(u,u)}$ on the subspace
$V$ to obtain $\lambda_j(\omega)$ and solving the equation $\lambda_j(\omega)=\omega^2$.
Since $\omega$ is continuous and decreasing from $+\infty$ to $0$ in $\ep_1,\mu_1$ seperately, one
can choose the material parameters such that
$\epsilon_0\mu_0\omega^2-(M+1+\kappa)^2<0<\epsilon_0\mu_0\omega^2-(M+\kappa)^2$, i.e.
for any pair $(\kappa,\omega)$
there are $2M+1$ values $-M,-M+1,\ldots, M-1, M$ of $m$ corresponding to propagating harmonics.

The field $u$ obtained in the space $V$ is automatically a guided mode, as the propagating harmonics automatically vanish
in the subspace $V$.
\end{proof}

As an example, if we let $M=N=0$, then $2M+N+2=2$, $2M+1=1$, $N+1=1$,
The pillar has period $\pi$ and $\epsilon_{2j+1}=0$ for all $j$, and we can allow
  one propagating harmonic.
We apply the min-max principle on the space $V=\{u\in\Honek:u_{2j}=0, \forall j\}$
and by choosing proper $\ep_1$ we can obtain an eigenfunction of smallest period $2\pi$ that
is automatically a guided mode.

If we take $M=1$, $N=0$, then $2M+N+2=4$, $2M+1=3$ and $N+1=1$.
Let $\epsilon,\mu$ have period $\pi/2$ and
so $\ep_j=0$ for $j\not\in4\ZZ$, or say $\forall j$,
and we can allow to have up to $2M+1=3$ propagating harmonics.
One can minimize the Rayleigh quotient on the space $V=\{u\in\Honek:u_{4j-1}=u_{4j}=u_{4j+1}=0, \forall j\}$.
If we take $M=N=1$, then $2M+N+2=5$, $2M+1=3$, and $N+1=2$.
The parameters $\ep$ and $\mu$ have period $2\pi/5$ and can be
allowed to have up to $2M+1=3$ propagating harmonics.
We apply the min-max principle on the space $V=\{u\in\Honek:u_{5j+1}=u_{5j+2}=u_{5j+3}=u_{5j+4}=0, \forall j\}$.
The pseudo-period of the embedded guided mode is $2\pi$.

In our design, the wave number $\kappa$ can be nonzero and
 there exists a continuous embedded dispersion relation $\omega(\kappa)$. The guided
mode is robust with respect to $\kappa$.
It is also noticed that the modes are subject to the periodicity $\frac{2\pi}{2M+N+2}$.
If the material is perturbed in a say that destroys he small period
lose this periodicity while retaining the period
$2\pi$, the guided mode vanishes.

This design can also be understood as an existence proof of a guided mode with a larger
pseudo-periodicity.
If we assume the smallest period of the pillar is
$2\pi$, embedded guided modes with period $(2M+N+2)2\pi$ can exist.

\subsection{Nonexistence}
Nonexistence results for periodic slabs can be found in \cite{ShipmanVolkov2007}\cite{Bonnet-BeStarling1994}.
In \cite{ShipmanVolkov2007}, the nonexistence of guided modes in inverse structures is discussed.
Consider the piecewise constant material as in Theorem \ref{thm:EigenPillar}.
An inverse structure is a periodic structure with the material parameters $\epsilon_1,\mu_1$
less than the corresponding parameters $\ep_0,\mu_0$ in the exterior of the material.
The proof in \cite{ShipmanVolkov2007} requires that the slab satisfy a certain restriction condition.
The proof of the nonexistence includes introducing the subspace $X$ in which the propagating and linear
harmonics vanish then estimating the minimum of the Rayleigh quotient.
With the restriction of the slab width, it is shown that the Rayleigh quotient
is strictly bounded below by in inverse structures, and hence the weak problem has no solution in $X$.
We use an analogous restriction on the radius of the pillar in our proof, and whether
this restriction is necessary remains an open problem.

In \cite{Bonnet-BeStarling1994}, the assumption is on the parameters only.
It is assumed that there exists one plane parallel to the slab such
that the material parameters $\ep,\mu$ are nondecreasing in the direction
perpendicular to the slab.
In Theorem \ref{thm:Nonexistence2}, we present an analogous
condition that the material parameters be nondecreasing in the radial direction.
The proof involves an appropriate Rayleigh identity.

%


\begin{theorem}
(Nonexistence of guided modes)
Assume that in $\Omega$, $\ep_-<\ep\le\ep_0$ and $\mu_-<\mu\le\mu_0$.
Let the frequency $\omega$ and the wave number $\kappa$ be given
in the first Brillouin zone $[-\frac{1}{2},\frac{1}{2})$.
Suppose that the radius $R$ of the pillar satisfies
\begin{equation}  \label{cond:R_small}
R<\frac{1}{\sqrt{\epsilon_0\mu_0\omega^2-\kappa^2}}
\end{equation}
Then the periodic pillar does not admit any guided modes at the given frequency and wavenumber.
\end{theorem}
\begin{proof}
We restrict to the subspace $X\subset\Honek$ with
\[
X=\{u\in\Honek: \int_{\Gamma_R} u(x,y,z)e^{-i\ell\theta}e^{-i(m+\kappa)z}=0,
       \mbox{ if either } m\in\mathcal{Z}_p \mbox{, or } m\in\mathcal{Z}_a \mbox{ and } \ell=0\}
\]
The form $a^\omega(\cdot,\cdot)$ is conjugate symmetric in $X$, and the weak problem
\eqref{eqn:WeakReal1} is equivalent to
$a^\omega(u,v)-\omega^2b(u,v)=0 \mbox{ on } X$, as well as
$a^\omega(u,v)-\omega^2b(u,v)=0$ for all $v\in X^{\perp}$.
This gives rise to a finite number of extra conditions
$(\widehat{\partial_n u|_{\Gamma_R}})_{m\ell}=0, \forall m\in\mathcal{Z}_p\mbox{ or } m\in\mathcal{Z}_a,\ell=0$.

Consider the eigenvalue problem $a^\omega(u,v)-\alpha\omega^2 b(u,v)=0$ on $X$.
On $X$, $a^\omega(u,v)=a_e^\omega(u,v)$.
The problem of guided modes is solved by minimizing the quotient $\frac{a(u,u)}{b(u,u)}$
on $X$. 
Of course, the field $u$ should satisfy the following radiation condition:
\begin{equation}\label{cond:ExtraX}
\widehat{(\partial_n u|_{\Gamma_R})}_{m\ell}+\gamma_{m\ell}\widehat{(u|_{\Gamma_R})}_{m\ell}=0,
    \quad \forall m\in\mathcal{Z}_e \mbox{ or } m\in\mathcal{Z}_a \mbox{ and } \ell\neq0.
\end{equation}
We first let $\epsilon_1=\epsilon_0$, $\mu_1=\mu_0$.
The eigenfunctions satisfy a strong form
of the Helmholtz equation
\begin{equation}
\left\{
\begin{array}{l}
(\nabla+i\boldsymbol\kappa)^2\psi+\alpha\epsilon_0\mu_0\omega^2\psi=0 \mbox{ in } \Omega_R \\
\psi\in X , \quad T\psi+\partial_n\psi|_\Gamma=0 \\
\psi \mbox{ satisfies pediodic boundary conditions in } X.
\end{array}
\right.
\end{equation}
In $\Omega_R$, the separable solutions are in the form of
\begin{equation}
\begin{split}
&A_{m\ell}J_\ell(|\zeta_m| r)e^{i\ell\theta}e^{i(m+\kappa)z} \mbox{, if } \zeta_m^2>0, \\
&A_{m\ell}I_\ell(|\zeta_m| r)e^{i\ell\theta}e^{i(m+\kappa)z} \mbox{, if } \zeta_m^2<0,\\
&\left[C_{m1}+C_{m2}\ln|r|\right] e^{i\ell\theta}e^{i(m+\kappa)z}
   \mbox{, if } \zeta_m^2=0, \mbox{ and } \ell=0, \\
&\left[C_{m\ell1}|r|^\ell+C_{m\ell2}|r|^{-\ell}\right]e^{i\ell\theta}e^{i(m+\kappa)z}
   \mbox{, if } \zeta_m^2=0 \mbox{ and } \ell\not=0,
\end{split}
\end{equation}
where $\zeta_m^2=\alpha\ep_0\mu_0\omega^2-(m+\kappa)^2$.

We treat the cases for $m$ separately.

Case I: $ m\in\mathcal{Z}_p$, i.e. $\eta_m^2>0$. In this case,
the propagating harmonics should vanish, and $\widehat{(u|_{\Gamma_R})}_{m\ell}=0$.
If $\zeta_m^2>0$, and we assume $\zeta_m>0$, then
\[
J_\ell(|\zeta_m| R)=0,
\]
so $j_\ell=\zeta_m R=\sqrt{\alpha\epsilon_0\mu_0-(m+k)^2}R$,
where $j_\ell$ is a zero of $J_\ell(x)$.
The eigenvalues are given by
\[
\alpha=\frac{\frac{j_\ell^2}{R^2}+(m+k)^2}{\epsilon_0\mu_0\omega^2}
\]
The Bessel function $J_l(z)$ has a sequence of zeros, and the corresponding
$\alpha$ form a sequence of eigenvalues $\{\alpha_j\}_{j=1}^\infty$
with all possible $j_l$ and $m\in\mathbb{Z}$.
According to our assumption of the radius of the pillar, the eigenvalues
\[
\begin{split}
\alpha_{m\ell} &=\frac{1}{\epsilon_0\mu_\omega^2}\left[\frac{j_\ell^2}{R^2}+(m+\kappa)^2\right]\\
    &\ge\frac{1}{\epsilon_0\mu_0\omega^2}\left[\frac{j_\ell^2}{R^2}+\kappa^2\right] \\
    &\ge\frac{1}{\epsilon_0\mu_0\omega^2}\left[\frac{\ell^2}{R^2}+\kappa^2\right] \\
    &\ge\frac{1}{\epsilon_0\mu_0\omega^2}\left[\frac{1}{R^2}+\kappa^2\right] \\
    &> 1
\end{split}
\]
If $\zeta_m^2=0$, the pillar does not support such harmonics for $\ell=0$.
For $\ell\neq0$, the separable solution $C_{m\ell_1}r^\ell+C_{m\ell_2}r^{-\ell}$ should satisfy
$C_{m\ell_1}R^\ell+C_{m\ell_2}R^{-\ell}=0$ which is not possible.
If $\zeta_m^2<0$, we assume $\zeta_m=i|\zeta_m|$ and
\[
I_\ell(|\zeta_m|R)=0.
\]
It is not possible since the modified Bessel functions $I_\ell$ have no
zeros except at 0.

Case II: $m\in\mathcal{Z}_e$, i.e. $\eta_m^2<0$, and $\eta_m=i|\eta_m|$.
In this case, the conditions in \eqref{cond:ExtraX} for $m$ should be satisfied.
If $\zeta_m^2>0$, and we assume $\zeta_m>0$, then
\[
\frac{d}{dr}J_\ell(\zeta_mr)|_{r=R}=-\gamma_{ml}J_\ell(\zeta_m R)
\]
where $\gamma_{m\ell}=-\frac{\eta_m H_\ell^{1'}(\eta_mR)}{H_\ell^1(\eta_mR)}$.
The value of $R$ can be solved, and by comparing $\zeta_m^2$ and $\eta_m^2$, one knows that $\alpha>1$.
If $\zeta_m^2=0$, we also have $\alpha>1$.
If $\zeta_m^2<0$, we assume $\zeta_m=i|\zeta_m|$.
Then
\[
|\zeta_m| I_\ell'(|\zeta_m| R)+\gamma_{ml}I_\ell(|\zeta_m|R)=0.
\]
However we know that $I'_\ell(|\zeta_m| R)>0$, $\gamma_{m\ell}>0$ and $I_\ell(|\zeta_m| R)>0$,
and consequently the left hand side cannot be 0.

Case III: $\eta_m^2=0$ and $\ell\neq0$. The condition
$\widehat{(\partial_n u|_{\Gamma_R})}_{m\ell}+\gamma_{m\ell}\widehat{(u|_{\Gamma_R})}_{m\ell}=0$
should be satisfied.
If $\zeta_m^2>0$, then $\alpha>1$.
If $\zeta_m^2=0$, there is no solution.
If $\zeta_m^2<0$, $|\zeta_m| I_\ell'(|\zeta_m| R)+\gamma_{ml}I_\ell(\zeta_mR)=0$.
It is not possible.

Case IV: $\eta_m^2=0$ and $\ell=0$. The guided modes satisfy
$\widehat{(u|_{\Gamma_R})}_{m\ell}=0$.
If $\zeta_m^2>0$, then $\alpha>1$.
If $\zeta_m^2=0$, there is no solution.
If $\zeta_m^2<0$, $I_\ell(|\zeta|_m R)=0$.
It is not possible because the Bessel function $I_0$ has no zero.

In general, when $\epsilon_1=\epsilon_0$, any eigenvalue
$\alpha>1$.

In the case when
 $\mu_-<\mu\le\mu_0$ and $\ep_-<\epsilon\le\epsilon_0$,
 the quotient $\frac{a_r(u,u)}{b(u,u)}$ with coefficient
 $\ep,\mu$ is greater than or equal to the quotient with
quotient $\ep_0,\mu_0$.
Under our assumption of the size of the pillar, the number $\alpha>1$.
As a result, there exists no guided mode because the number $\alpha>1$ does not correspond
to a guided mode.
\end{proof}


\begin{theorem}\label{thm:Nonexistence2}
Assume there is a pair $x_0,y_0$ such that for all $-\pi\le z\le\pi$ and any
vector $\boldsymbol{r}_0=(r_{0x},r_{0y},0)$, the material parameters
$\epsilon, \mu$ are nondecreasing along the direction of $\boldsymbol{r_0}$,
that is, the weak directional derivatives $\nabla\epsilon\cdot\boldsymbol{r_0}$, and
$\nabla\mu\cdot\boldsymbol{r_0}$ are nonnegative.
Then there exists no guided mode.
\end{theorem}
\begin{proof}
Using polar coordinates, we observe that
\[
\begin{split}
\nabla\cdot(r\frac{\partial u}{\partial r}\mu^{-1}\nabla\bar u)=&
  \nabla(r u_r)\cdot\mu^{-1}\nabla\bar u+r u_r(\nabla\cdot\mu^{-1}\nabla \bar u)\\
  =&r u_r(\nabla\cdot\mu^{-1}\nabla\bar u)+u_r(\nabla r\cdot\mu^{-1}\nabla\bar u)+r\nabla u_r\cdot\mu^{-1}\nabla\bar u.
\end{split}
\]
We integrate this to obtain
\[
\begin{split}
\int_{\Gamma_R}ru_r\mu_0^{-1}\frac{\partial\bar u}{\partial n}
   &=\int_{\Omega_R}ru_r(\nabla\cdot\mu^{-1}\nabla\bar u)+\int_{\Omega_R}u_r(\nabla r\cdot\mu^{-1}\nabla\bar u)
      +\int_{\Omega_R}r\nabla u_r\cdot\mu^{-1}\nabla\bar u\\
   &=-\omega^2\int_{\Omega_R}r\epsilon u_r\bar u+\int_{\Omega_R} u_r(\nabla r\cdot\mu^{-1}\nabla\bar u)
      +\int_{\Omega_R}r\nabla u_r\cdot\mu^{-1}\nabla\bar u.
\end{split}
\]
Adding its complex conjugate, we have
\[
2\int_{\Gamma_R}\mu_0^{-1}R|\frac{\partial u}{\partial r}|^2=
   -\omega^2\int_{\Omega_R}\ep r\frac{\partial}{\partial r}|u|^2
   +\int_{\Omega_R}u_r(\nabla r \cdot\mu^{-1}\nabla\bar u)
   +\int_{\Omega_R}\bar u_r(\nabla r \cdot\mu^{-1}\nabla u)
   +\int_{\Omega_R}\mu_{-1}r\frac{\partial}{\partial r}|\nabla u|^2.
\]
We integrate by parts in $r$ for terms including $r\frac{\partial}{\partial r}$,
\[
\begin{split}
\int_{\Omega_R}\ep r \frac{\partial |u|^2}{\partial r}
   =&\int_0^{2\pi}\int_{-\pi}^\pi\int_0^R \ep r\frac{\partial |u|^2}{\partial r}rdrdzd\theta\\
   =&\int_0^{2\pi}\int_{-\pi}^\pi\int_0^R \epsilon r^2\frac{\partial|u|^2}{\partial r} drdzd\theta\\
   =&\int_0^{2\pi}\int_{-\pi}^\pi\ep r^2|u|^2|_0^Rdzd\theta\
       -\int_{\Omega_R}2\ep r|u|^2drdzd\theta
       -\int_{\Omega_R} r^2\frac{\partial\ep}{\partial r}|u|^2drdzd\theta\\
   =&\int_0^{2\pi}\int_{-\pi}^\pi\ep R^2|u|^2|_0^Rdzd\theta\
       -\int_{\Omega_R}2\ep |u|^2
       -\int_{\Omega_R} r\frac{\partial\ep}{\partial r}|u|^2, 
\end{split}
\]
and
\[
\int_{\Omega_R}\mu^{-1} r \frac{\partial |u|^2}{\partial r}
   =\int_0^{2\pi}\int_{-\pi}^\pi\mu^{-1} R^2|u|^2|_0^Rdzd\theta\
       -\int_{\Omega_R}2\mu^{-1} |u|^2
       -\int_{\Omega_R} r\frac{\partial\mu^{-1}}{\partial r}|u|^2.
\]
The previous identity becomes
\[
\begin{split}
2\int_{\Gamma_R}\mu_0^{-1}R|\frac{\partial u}{\partial r}|^2
   =&-\omega^2\left[\int_0^{2\pi}\int_{-\pi}^\pi R^2\epsilon_0|u(R)|^2dzd\theta
                   -\int_{\Omega_R}2\epsilon|\nabla u|^2-\int_{\Omega_R}r\frac{\partial\ep}{\partial r}|u|^2 \right]\\
    &+\int_{\Omega_R}u_r(\nabla r\cdot \mu^{-1}\nabla\bar u)+\int_{\Omega_R}\bar u_r(\nabla r\cdot\mu^{-1}\nabla u)\\
    &+\left[\int_0^{2\pi}\int_{-\pi}^\pi R^2\mu_0^{-1}|u(R)|^2dzd\theta
                   -\int_{\Omega_R}2\mu^{-1}|\nabla u|^2-\int_{\Omega_R}r\frac{\partial\mu^{-1}}{\partial r}|u|^2 \right]
\end{split}
\]
Since the field satisfies the Helmholtz equation, we can replace $-\int_{\Omega_R}\mu^{-1}|\nabla u|^2$
by $-\omega^2\int_{\Omega_R}\ep|u|^2+\mu_0^{-1}\int_{\Gamma_R}\bar u Tru$ to obtain
\[
\begin{split}
2\int_{\Gamma_R}\mu_0^{-1}R|\frac{\partial u}{\partial r}|^2
   =&\left[-\omega^2\int_0^{2\pi}\int_{-\pi}^\pi R^2\epsilon_0|u(R)|^2dzd\theta
                   +\omega^2\int_{\Omega_R}2\epsilon|\nabla u|^2
                   +\omega^2\int_{\Omega_R}r\frac{\partial\ep}{\partial r}|u|^2 \right]\\
    &+\int_{\Omega_R}u_r(\nabla r\cdot \mu^{-1}\nabla\bar u)+\int_{\Omega_R}\bar u_r(\nabla r\cdot\mu^{-1}\nabla u)\\
    &+\left[\int_0^{2\pi}\int_{-\pi}^\pi R^2\mu_0^{-1}|u(R)|^2dzd\theta
                   -2\omega^2\int_{\Omega_R}\ep|u|^2+2\mu_0^{-1}\int_{\Gamma_R}\bar u Tru
                   -\int_{\Omega_R}r\frac{\partial\mu^{-1}}{\partial r}|u|^2 \right],
\end{split}
\]
and so
\[
\begin{split}
2\int_{\Gamma_R}\mu_0^{-1}R|\frac{\partial u}{\partial r}|^2
   +&\omega^2\int_0^{2\pi}\int_{-\pi}^\pi R^2\epsilon_0|u(R)|^2dzd\theta
   -\int_0^{2\pi}\int_{-\pi}^\pi R^2\mu_0^{-1}|u(R)|^2dzd\theta\\
   =&\omega^2\int_{\Omega_R}r\frac{\partial\ep}{\partial r}|u|^2+\int_{\Omega_R}u_r(\nabla r\cdot \mu^{-1}\nabla\bar u)
    +\int_{\Omega_R}\bar u_r(\nabla r\cdot\mu^{-1}\nabla u)\\
    &-\int_{\Omega_R}r\frac{\partial\mu^{-1}}{\partial r}|u|^2+2\mu_0^{-1}\int_{\Gamma_R}\bar u Tru.
\end{split}
\]
In this identity,
\[
u_r(\nabla r\cdot\mu^{-1}\nabla\bar u)=\mu^{-1}u_r(\boldsymbol r\cdot\nabla\bar u)
=(u_r\boldsymbol r)\cdot\nabla\bar u \mu^{-1}
=\left|\frac{\partial u}{\partial r}\boldsymbol r\right|^2\mu^{-1},
\]
where
$\nabla u=\frac{\partial u}{\partial z}\boldsymbol z
     +\frac{\partial u}{\partial r}\boldsymbol r
     +\frac{1}{r}\frac{\partial u}{\partial\theta}\boldsymbol\theta$,
$\boldsymbol{r}=(\cos\theta,\sin\theta,0), \boldsymbol\theta=(-\sin\theta,\cos\theta,0),\boldsymbol z=(0,0,1)$,
and
$|\nabla u|^2=|u_r|^2+\frac{1}{r^2}|u_\theta|^2+|u_{z}|^2$.
Simplify it to obtain
\begin{equation} \label{idty:Unicity}
\begin{split}
\omega^2\int_{\Omega_R}r\frac{\partial\ep}{\partial r}|u|^2
  &+2\int_{\Omega_R}\mu^{-1}|\frac{\partial u}{\partial r}\boldsymbol r|^2
  -\int_{\Omega_R}r\frac{\partial \mu^{-1}}{\partial r}|\nabla u|^2
  +2\mu_0^{-1}\int_{\Gamma_R}\bar u Tru\\
  =&2\int_{\Gamma_R}\mu_0^{-1}R|\frac{\partial u}{\partial r}|^2
   +\omega^2\int_0^{2\pi}\int_{-\pi}^\pi R^2\epsilon_0|u(R)|^2dzd\theta
   -\int_0^{2\pi}\int_{-\pi}^\pi R^2\mu_0^{-1}|\nabla u(R)|^2dzd\theta
\end{split}
\end{equation}

The left-hand side of the identity \eqref{idty:Unicity} is nonnegative by our condition
on the material parameters, and it vanishes if and only if $\|u\|_{\Honek}=0$.
If we assume $u$ is a guided mode, and $u$ has the expansion
\[
u(r,\theta,z)=\sum_{m\in \mathcal{Z}_e}\sum_\ell a_{m\ell} H_\ell^1(\eta_m r) e^{i\ell\theta}e^{i(m+k)z}
   +\sum_{m\in \mathcal{Z}_a}\sum_{\ell\neq0} cr^{-|\ell|}e^{i\ell\theta}e^{i(m+k)z},
\]
then the terms with $m\in\mathcal{Z}_a$ of the right hand side of \eqref{idty:Unicity}  is a sum of multiples of
\[
\omega^2\epsilon_0R^{-2|\ell|+2}
   -\mu_0^{-1}(m+\kappa)^2R^{-2|\ell|+2}=0.
\]
%

Since $H_\ell^{1}(\eta_mR)$ and $H_\ell^{1'}(\eta_mR)$ are exponentially decaying
as $R\rightarrow\infty$,
in this limit, the limit of the right hand side is $0$.
On the other hand, the left hand side does not converge to $0$ if $u\ne0$.
 So $u=0$.
\end{proof}

\bibliographystyle{plain}
\bibliography{Tu_Pillar}

\end{document}